\documentclass[hidelinks,11pt]{article}

\usepackage{fullpage}

\usepackage{amsmath}
\usepackage{amssymb}
\usepackage{amsfonts,amsthm}
\usepackage{thmtools,thm-restate}
\usepackage{tabularx,lipsum,environ}
\usepackage{multirow}
\usepackage{hyperref}
\usepackage{complexity}
\usepackage{enumerate}
\usepackage{graphicx,float}
\usepackage{subcaption}

\usepackage{authblk}

\usepackage{boxedminipage}
\usepackage{framed}
\usepackage{thmtools}
\usepackage{thm-restate}
\usepackage{xspace}
\usepackage{xifthen}
 \usepackage{tabularx}
\usepackage{tikz} 
 \usetikzlibrary{calc}

\usepackage{todonotes}

\newtheorem{thm-main}{Theorem}
\newtheorem{theorem}{Theorem}[section]
\newtheorem{lemma}[theorem]{Lemma}
\newtheorem{corollary}[theorem]{Corollary}

\newtheorem{remark}[theorem]{Remark}

\newtheorem{proposition}[theorem]{Proposition}
\newtheorem{observation}[theorem]{Observation}

\theoremstyle{definition}
\newtheorem{definition}[theorem]{Definition}
\newtheorem{reductionrule}[theorem]{Reduction Rule}

\newtheorem{openproblem}{Open Problem}

\newtheorem{nestedclaim}{Claim}[theorem]

\newtheorem{nestedobservation}[nestedclaim]{Observation}

\newcommand{\bprob}{$b$-\textsc{Chromatic Number}}
\newcommand{\bcol}{$b$-\textsc{Coloring}}

\newenvironment{claimproof}{\begin{proof}\renewcommand{\qedsymbol}{\claimqed}}{\end{proof}\renewcommand{\qedsymbol}{\plainqed}}
\let\plainqed\qedsymbol

\makeatletter
\renewcommand{\paragraph}{%
  \@startsection{paragraph}{4}%
  {\z@}{1ex \@plus 1ex \@minus .2ex}{-1em}
  {\normalfont\normalsize\bfseries}%
}
\makeatother

\usepackage{tikz}

\newlength{\RoundedBoxWidth}
\newsavebox{\GrayRoundedBox}
\newenvironment{GrayBox}[1]%
   {\setlength{\RoundedBoxWidth}{.93\textwidth}
    \def\boxheading{#1}
    \begin{lrbox}{\GrayRoundedBox}
       \begin{minipage}{\RoundedBoxWidth}}%
   {   \end{minipage}
    \end{lrbox}
    \begin{center}
    \begin{tikzpicture}%
       \node(Text)[draw=black!20,fill=white,rounded corners,%
             inner sep=2ex,text width=\RoundedBoxWidth]%
             {\usebox{\GrayRoundedBox}};
        \coordinate(x) at (current bounding box.north west);
        \node [draw=white,rectangle,inner sep=3pt,anchor=north west,fill=white] 
        at ($(x)+(6pt,.75em)$) {\boxheading};
    \end{tikzpicture}
    \end{center}}     

\newenvironment{defproblemx}[2][]{\noindent\ignorespaces%
                                \FrameSep=6pt%
                                \parindent=0pt%
                \vspace*{-1.5em}
                \ifthenelse{\isempty{#1}}{%
                  \begin{GrayBox}{\textsc{#2}}%
                }{%
                  \begin{GrayBox}{\textsc{#2} parameterized by~{#1}}%
                }
                \newcommand\Input{Input:}%
                \begin{tabular*}{\textwidth}{@{\hspace{.1em}} >{\itshape} p{1.8cm} p{0.8\textwidth} @{}}%
            }{
                \end{tabular*}%
                \end{GrayBox}%
                \ignorespacesafterend
            }

\newcommand{\defproblema}[3]{
  \begin{defproblemx}{#1}
    Input:  & #2 \\
    Question: & #3
  \end{defproblemx}
}%

\newcommand\problemdef[3]{
	\defproblema{#1}{#2}{#3}
}

\newcommand\bN{\mathbb{N}}

\newcommand\cO{\mathcal{O}}

\newcommand*{\defeq}{\mathrel{\vcenter{\baselineskip0.5ex \lineskiplimit0pt
                     \hbox{\scriptsize.}\hbox{\scriptsize.}}}%
                     =}

\newcommand\yes{\textsc{Yes}}
\newcommand\no{\textsc{No}}

\newcommand\dist{dist}
\newcommand\coloring{\gamma}
\newcommand\coloringg{\delta}

\newcommand\card[1]{\left|#1\right|}

\usepackage[linesnumbered,lined,boxed,noend]{algorithm2e}

\newcommand\cC{\mathcal{C}}

\title{A Complexity Dichotomy for Critical Values of the $b$-Chromatic Number of Graphs}

\author[1]{Lars Jaffke\thanks{Supported by the Bergen Research Foundation (BFS).}} 
\author[1]{Paloma T.\ Lima\thanks{Supported by Research Council of Norway via the project ``CLASSIS''.}}

\affil[1]{Department of Informatics\\
	University of Bergen\\
	Bergen, Norway}
\affil[ ]	{\texttt{\{lars.jaffke,paloma.lima\}@uib.no}}

\date{}

\begin{document}

\maketitle

\begin{abstract}
	A \emph{$b$-coloring} of a graph $G$ is a proper coloring of its vertices such that each color class contains a vertex that has at least one neighbor in all the other color classes. The \bcol{} problem asks whether a graph $G$ has a $b$-coloring with $k$ colors.
	The \emph{$b$-chromatic number} of a graph $G$, denoted by $\chi_b(G)$, is the maximum number $k$ such that $G$ admits a $b$-coloring with $k$ colors. We consider the complexity of the \bcol{} problem, whenever the value of $k$ is close to one of two upper bounds on $\chi_b(G)$: The maximum degree $\Delta(G)$ plus one, and the $m$-degree, denoted by $m(G)$, which is defined as the maximum number $i$ such that $G$ has $i$ vertices of degree at least $i-1$. We obtain a dichotomy result stating that for fixed $k \in \{\Delta(G) + 1 - p, m(G) - p\}$, the problem is polynomial-time solvable whenever $p \in \{0, 1\}$ and, even when $k = 3$, it is \NP{}-complete whenever $p \ge 2$.
	We furthermore consider parameterizations of the \bcol{} problem that involve the maximum degree $\Delta(G)$ of the input graph $G$ and give two \FPT{}-algorithms. First, we show that deciding whether a graph $G$ has a $b$-coloring with $m(G)$ colors is \FPT{} parameterized by $\Delta(G)$. Second, we show that \bcol{} is \FPT{} parameterized by $\Delta(G) + \ell_k(G)$, where $\ell_k(G)$ denotes the number of vertices of degree at least $k$.
\end{abstract}

\section{Introduction}
Given a set of colors, a \emph{proper coloring} of a graph is an assignment of a color to each of its vertices in such a way that no pair of adjacent vertices receive the same color.
In the deeply studied \textsc{Graph Coloring} problem, we are given a graph and the question is to determine the smallest set of colors with which we can properly color the input graph.
This problem is among Karp's famous list of 21 \NP{}-complete problems~\cite{Kar72} and since it often arises in practice, heuristics to solve it are deployed in a wide range of applications. A very natural such heuristic is the following. We greedily find a proper coloring of the graph, and then try to \emph{suppress} any of its colors in the following way: say we want to suppress color~$c$. If there is a vertex $v$ that has received color $c$, and there is another color $c' \neq c$ that does not appear in the neighborhood of $v$, then we can safely recolor the vertex~$v$ with color~$c'$ without making the coloring improper. We terminate this process once we cannot suppress any color anymore.

To predict the worst-case behavior of the above heuristic, Irving and Manlove defined the notions of a \emph{$b$-coloring} and the \emph{$b$-chromatic number} of a graph~\cite{IM99}. A $b$-coloring of a graph $G$ is a proper coloring such that in every color class there is a vertex that has a neighbor in all of the remaining color classes, and the $b$-chromatic number of $G$, denoted by $\chi_b(G)$, is the maximum integer $k$ such that $G$ admits a $b$-coloring with $k$ colors. We observe that in a $b$-coloring with $k$ colors, there is no color that can be suppressed to obtain a proper coloring with $k-1$ colors, hence $\chi_b(G)$ describes the worst-case behavior of the previously described heuristic on the graph $G$.
We consider the following two computational problems associated with $b$-colorings of graphs.

\problemdef
	{$b$-Coloring}
	{Graph $G$, integer $k$}
	{Does $G$ admit a $b$-coloring with $k$ colors?}
\problemdef
	{$b$-Chromatic Number}
	{Graph $G$, integer $k$}
	{Is $\chi_b(G) \ge k$?}
We would like to point out an important distinction from the `standard' notion of proper colorings of graphs: If a graph $G$ has a $b$-coloring with $k$ colors, then this implies that $\chi_b(G) \ge k$. However, if $\chi_b(G) \ge k$ then we can in general not conclude that $G$ has a $b$-coloring with $k$ colors. A graph for which the latter implication holds as well is called \emph{$b$-continuous}. This notion is mostly of structural interest, since the problem of determining if a graph is $b$-continous is \NP-complete even if an optimal proper coloring and a $b$-coloring with $\chi_b(G)$ colors are given~\cite{BCF07}.

Besides observing that $\chi_b(G) \le \Delta(G) + 1$ where $\Delta(G)$ denotes the maximum degree of $G$, Irving and Manlove~\cite{IM99} defined the \emph{$m$-degree} of $G$ as the largest integer $i$ such that $G$ has $i$ vertices of degree at least $i - 1$. It follows that $\chi_b(G) \le m(G)$. Since the definition of the $b$-chromatic number originated in the analysis of the worst-case behavior of graph coloring heuristics, graphs whose $b$-chromatic numbers take on \emph{critical} values, i.e.\ values that are close to these upper bounds, are of special interest. In particular, identifying them can be helpful in structural investigations concerning the performance of graph coloring heuristics.

In terms of computational complexity, Irving and Manlove showed that both \bcol{} and \bprob{} are \NP{}-complete~\cite{IM99} and Sampaio observed that \bcol{} is \NP{}-complete even for every fixed integer $k$~\cite{Sam12}. Panolan et al.~\cite{PPS17} gave an exact exponential algorithm for \bprob{} running in time $\cO(3^nn^4\mbox{ log }n)$ and an algorithm that solves \bcol{} in time $\cO(\binom{n}{k} 2^{n-k} n^4\log n)$.
From the perspective of parameterized complexity~\cite{CyganEtAl15,DF13}, it has been shown that \bprob{} is $\W[1]$-hard parameterized by~$k$~\cite{PPS17} and that the dual problem of deciding whether $\chi_b(G) \ge n - k$, where $n$ denotes the number of vertices in~$G$, is \FPT{} parameterized by $k$~\cite{HS13}.

Since the above mentioned upper bounds $\Delta(G) + 1$ and $m(G)$ on the $b$-chromatic number are trivial to compute, it is natural to ask whether there exist efficient algorithms that decide whether $\chi_b(G) = \Delta(G) + 1$ or $\chi_b(G) = m(G)$. It turns out both these problems are \NP{}-complete as well~\cite{HSS12,IM99}. However, it is known that the problem of deciding whether a graph $G$ admits a $b$-coloring with $k = \Delta(G) + 1$ colors is \FPT{} parameterized by $k$~\cite{PPS17,Sam12}.

\paragraph*{The Dichotomy Result.} One of the main contributions of this paper is a complexity dichotomy of the \bcol{} problem for fixed $k$, whenever $k$ is close to either $\Delta(G) + 1$ or $m(G)$. In particular, for fixed $k \in \{\Delta(G) + 1 - p, m(G) - p\}$, we show that the problem is polynomial-time solvable when $p \in \{0, 1\}$ and, even in the case $k = 3$, \NP{}-complete for all fixed $p \ge 2$. More specifically, we give \XP{} time algorithms for the cases $k = m(G)$, $k = \Delta(G)$, and $k = m(G) - 1$ which together with the \FPT{} algorithm for the case $k = \Delta(G) + 1$~\cite{PPS17,Sam12} and the aforementioned \NP{}-hardness result for $k = 3$ completes the picture. 
We now formally state this result.
\begin{thm-main}\label{thm:main}
	Let $G$ be a graph, $p \in \bN$ and $k \in \{\Delta(G) + 1 - p, m(G) - p\}$. The problem of deciding whether $G$ has a $b$-coloring with $k$ colors is
	\begin{enumerate}[(i)]
		\item \NP{}-complete if $k$ is part of the input and $p \in \{0, 1\}$,\label{thm:main:np-c:01}
		\item \NP{}-complete if $k = 3$ and $p \ge 2$, and\label{thm:main:np-c:2+}
		\item polynomial-time solvable for any fixed positive $k$ and $p \in \{0, 1\}$.\label{thm:main:poly}
	\end{enumerate}
\end{thm-main}

\paragraph*{Maximum Degree Parameterizations.} We consider parameterizations of \bcol{} that involve the maximum degree of the input graph and provide two new results. 
First, we show that deciding whether a graph admits a $b$-coloring with $m(G)$ colors is \FPT{} parameterized by $\Delta(G)$.
Second, as an extension of the first result, we prove that \bcol{} is \FPT{} parameterized by $\Delta(G) + \ell_k(G)$, where $\ell_k(G)$ denotes the number of vertices of degree at least $k$.
From the result of Kratochv\'{i}l et al.~\cite{KTV02}, stating that \bcol{} is \NP{}-complete for $k = \Delta(G) + 1$, it follows that \bcol{} is \NP{}-complete when $\Delta(G)$ is unbounded and $\ell_k(G) = 0$. On the other hand, Theorem~\ref{thm:main}(\ref{thm:main:np-c:2+}) implies that \bcol{} is already \NP{}-complete when $k = 3$ and $\Delta(G) = 4$. Together, this rules out the possibility of \FPT{}- and even of \XP{}-algorithms for parameterizations by one of the two parameters alone, unless $\P = \NP$.
Note that parameterizations of graph coloring problems by the number of high degree vertices have previously been considered for vertex coloring~\cite{Marx17} and edge coloring~\cite{GLPR19}.

\paragraph*{Organization.} The rest of the paper is organized as follows. After giving preliminary definitions in Section~\ref{sec:preliminaries}, we present the hardness results in Section~\ref{sec:hardness}, the algorithmic results of the dichotomy in Section~\ref{sec:alg:dich}, and the algorithms for the maximum degree parameterizations in Section~\ref{sec:alg:max:deg}. We conclude in Section~\ref{sec:conclusion}.

\section{Preliminaries}\label{sec:preliminaries}
We use the following notation: For $k \in \bN$, $[k] \defeq \{1,\ldots, k\}$. 
For a function $f \colon X \to Y$ and $X' \subseteq X$, we denote by~$f|_{X'}$ the restriction of $f$ to $X'$ and by $f(X')$ the set $\{f(x)\mid x\in X'\}$. For a set $X$ and an integer $n$, we denote by $\binom{X}{n}$ the set of all size-$n$ subsets of $X$.

\medskip
\noindent\textbf{Graphs.} Throughout the paper a graph $G$ with vertex set $V(G)$ and edge set $E(G) \subseteq \binom{V(G)}{2}$ is finite and simple. We often denote an edge $\{u, v\} \in E(G)$ by the shorthand $uv$.
For graphs $G$ and $H$ we denote by $H \subseteq G$ that $H$ is a subgraph of $G$, i.e.\ $V(H) \subseteq V(G)$ and $E(H) \subseteq E(G)$. We often use the notation $n \defeq |V(G)|$. For a vertex $v \in V(G)$, we denote by $N_G(v)$ the \emph{open neighborhood} of $v$ in $G$, i.e.\ $N_G(v) = \{w \in V(G) \mid vw \in E(G)\}$, and by $N_G[v]$ the \emph{closed neighborhood} of $v$ in $G$, i.e.~$N_G[v] \defeq \{v\} \cup N_G(v)$. For a set of vertices $X \subseteq V(G)$, we let $N_G(X) \defeq \bigcup_{v \in X} N_G(v) \setminus X$ and $N_G[X] \defeq X \cup N_G(X)$.
When $G$ is clear from the context, we abbreviate `$N_G$' to `$N$'. The \emph{degree} of a vertex $v \in V(G)$ is the size of its open neighborhood, and we denote it by $\deg_G(v) \defeq \card{N_G(v)}$ or simply by $\deg(v)$ if $G$ is clear from the context. For an integer $k$, we denote by $\ell_k(G)$ the number of vertices of degree at least $k$ in $G$.

For a vertex set $X \subseteq V(G)$, we denote by $G[X]$ the subgraph \emph{induced} by $X$, i.e.~$G[X] \defeq (X, E(G) \cap \binom{X}{2})$. We furthermore let $G - X \defeq G[V(G) \setminus X]$ be the subgraph of $G$ obtained from removing the vertices in $X$ and for a single vertex $x \in V(G)$, we use the shorthand `$G - x$' for `$G - \{x\}$'.

A graph $G$ is said to be \emph{connected} if for any $2$-partition $(X, Y)$ of $V(G)$, there is an edge $xy \in E(G)$ such that $x \in X$ and $y \in Y$, and \emph{disconnected} otherwise. A \emph{connected component} of a graph $G$ is a maximal connected subgraph of $G$. A \emph{path} is a connected graph of maximum degree two, having precisely two vertices of degree one, called its \emph{endpoints}. The \emph{length} of a path is its number of edges. Given a graph $G$ and two vertices $u$ and $v$, the \emph{distance} between $u$ and $v$, denoted by $\dist_G(u, v)$ (or simply $\dist(u, v)$ if $G$ is clear from the context), is the length of the shortest path in $G$ that has $u$ and $v$ as endpoints.

A graph $G$ is a \emph{complete graph} if every pair of vertices of $G$ is adjacent. A set $C\subseteq V(G)$ is a clique if $G[C]$ is a complete graph. A set $S\subseteq V(G)$ is an independent set if $G[S]$ has no edges. A graph $G$ is a \emph{bipartite graph} if its vertex set can be partitioned into two independent sets. 
A bipartite graph with bipartition $(A, B)$ is a \emph{complete bipartite graph} if all pairs consisting of one vertex from $A$ and one vertex from $B$ are adjacent, and with $a = \card{A}$ and $b = \card{B}$, we denote it by $K_{a, b}$. A \emph{star} is the graph $K_{1,b}$, with $b\geq 2$, and we call \emph{center} the unique vertex of degree $b$ and \emph{leaves} the vertices of degree one.

\medskip
\noindent\textbf{Colorings.} Given a graph $G$, a map $\coloring\colon V(G) \to [k]$ is called a \emph{coloring of $G$ with $k$ colors}.
If for every pair of adjacent vertices, $uv \in E(G)$, we have that $\coloring(u) \neq \coloring(v)$, then the coloring $\coloring$ is called \emph{proper}. For $i \in [k]$, we call the set of vertices $u \in V(G)$ such that $\coloring(u) = i$ the \emph{color class $i$}. If for all $i \in [k]$, there exists a vertex $x_i \in V(G)$ such that 
\begin{enumerate}[(i)]
	\item $\coloring(x_i) = i$, and
	\item for each $j \in [k] \setminus \{i\}$, there is a neighbor $y \in N_G(x_i)$ of $x_i$ such that $\coloring(y) = j$,
\end{enumerate}
then $\coloring$ is called a \emph{$b$-coloring} of $G$. For $i \in [k]$, we call a vertex $x_i$ satisfying the above two conditions a \emph{$b$-vertex for color $i$}.

\medskip
\noindent\textbf{Parameterized Complexity.} Let $\Sigma$ be an alphabet. A \emph{parameterized problem} is a set $\Pi \subseteq \Sigma^* \times \bN$. A parameterized problem $\Pi$ is said to be \emph{fixed-parameter tractable}, or contained in the complexity class \FPT{}, if there exists an algorithm that for each $(x, k) \in \Sigma^* \times \bN$ decides whether $(x, k) \in \Pi$ in time $f(k)\cdot \card{x}^c$ for some computable function $f$ and fixed integer $c \in \bN$. A parameterized problem $\Pi$ is said to be contained in the complexity class \XP{} if there is an algorithm that for all $(x, k) \in \Sigma^* \times \bN$ decides whether $(x, k) \in \Pi$ in time $f(k)\cdot n^{g(k)}$ for some computable functions $f$ and $g$.

A \emph{kernelization algorithm} for a parameterized problem $\Pi \subseteq \Sigma^* \times \bN$ is a polynomial-time algorithm that takes as input an instance $(x, k) \in \Sigma^* \times \bN$ and either correctly decides whether $(x, k) \in \Pi$ or outputs an instance $(x', k') \in \Sigma^* \times \bN$ with $\card{x'} + k' \le f(k)$ for some computable function $f$ for which $(x, k) \in \Pi$ if and only if $(x', k') \in \Pi$. We say that $\Pi$ \emph{admits a kernel} if there is a kernelization algorithm for $\Pi$.

\section{Hardness Results}\label{sec:hardness}
In this section we prove the hardness results of our complexity dichotomy. First, we show that \bprob{} and \bcol{} are \NP{}-complete for $k = m(G) - 1 = \Delta(G)$, based on a reduction due to Havet et al.~\cite{HSS12} who showed \NP{}-completeness for the case $k = m(G)$. 

\begin{theorem}\label{thm:npc:delta:mG-1}
\bprob{} and \bcol{} are \NP{}-complete, even when $k=m(G)-1=\Delta(G)$.
\end{theorem}
\begin{proof}
As in the proof of Havet et al.~\cite{HSS12}, the reduction is from the \NP{}-complete problem \textsc{3-Edge Coloring} of 3-regular graphs, which takes as input a 3-regular graph $G$ and asks whether the edges of $G$ can be properly colored with three colors.

Given an instance~$G$ of \textsc{3-Edge Coloring}, an instance $H$ of \bprob{} and \bcol{} is constructed as follows. The graph $H$ has one vertex for each vertex of $G$, that we denote by $v_1,\ldots,v_n$, one vertex for each edge, that we denote by $u_1,\ldots,u_m$ and a set of $4n+13$ vertices that we denote by $S$. The edge set of $H$ is such that $H[\{v_1,\ldots,v_n\}]$ is a clique, $H[S]$ is the disjoint union of one copy of the complete bipartite graph~$K_{n,n+3}$ and two copies of $K_{2,n+3}$ and $v_iu_j$ is an edge if the edge corresponding to $u_j$ is incident to the vertex corresponding to $v_i$ in $G$. 
The constructed graph $H$ is such that $\Delta(H)=n+3$ and $H$ has $n+4$ vertices of degree $n+3$, which implies that $m(H)=n+4$. The difference to the construction used in~\cite{HSS12} is that instead of the three complete bipartite graphs mentioned above, the authors use three copies of the star $K_{1,n+2}$.
\begin{nestedclaim}\label{claim:twins}
A connected component of $H$ that is a complete bipartite graph can contain $b$-vertices of at most one color in any $b$-coloring of $H$ with at least $n+3$ colors.
\end{nestedclaim}
\begin{claimproof}
Consider a component of $H$ that induces a $K_{i,n+3}$, with $i\in\{2,n\}$. In any $b$-coloring with $k\geq n+3$ colors, only the vertices of degree at least $n + 2$, so in this case the vertices of degree $n+3$ can be $b$-vertices in $H$. If $x$ is a $b$-vertex for a given color, then the remaining $k-1$ colors appear on the vertices of~$N(x)$. We conclude that any other vertex of degree $n+3$ of this component will be assigned the same color as $x$.
\end{claimproof}
\noindent We prove that $H$ has a $b$-coloring with $k=n+3$ colors if and only if $G$ is a \yes{}-instance for \textsc{3-Edge Coloring} by using the same steps as in the proof of Theorem 3 of \cite{HSS12} and with the additional use of Claim~\ref{claim:twins}. This proves the \NP{}-completeness of \bcol{} when $k=m(G)-1=\Delta(G)$. Furthermore, we prove that $\chi_b(H)\geq n+3$ if and only if $H$ has a $b$-coloring with $n+3$ colors. This yields the analogous result for \bprob{}.

 First, assume that $G$ is a \yes{}-instance for \textsc{3-Edge Coloring}. Let $\coloring_E \colon E(G) \to [3]$ be a proper 3-edge coloring for $G$. We construct a $b$-coloring $\coloring_{H}$ for $H$ in the following way. For each $1\leq i\leq \card{E(G)}$, $\coloring_{H}(u_i)=\coloring_E(e_i)$ and each $1\leq j\leq n$, we let $\coloring_{H}(v_j)=j+3$. Note that since $\coloring_E$ is a 3-edge coloring for $G$, the vertices $v_1,\ldots,v_n$ in $H$ are $b$-vertices for the colors $4,\ldots,n+3$: Any vertex in $G$ is incident with $3$ edges since $G$ is $3$-regular, and since $\coloring_E$ is proper, each such edge receives a different color. Hence, for any vertex $v_i$, the colors $\{1, 2, 3\}$ appear on $N_H(v_i) \cap \{u_1, \ldots, u_{\card{E(G)}}\}$.
Now we can color the rest of the graph $H$ in such a way that each connected component that is a complete bipartite graph contains a $b$-vertex for one of the three remaining colors. 

Now we consider the other direction. We start by observing that Claim~\ref{claim:twins} implies that $H$ does not admit a $b$-coloring with $n+4=m(H)=\Delta(H)+1$ colors, since the set of vertices of degree $n+3$ can contain $b$-vertices for at most three colors in any such a $b$-coloring. This implies that $\chi_b(H)\geq m(H)-1=\Delta(H)$ if and only if $H$ has a $b$-coloring with $m(H)-1=\Delta(H)$ colors.

Assume $H$ has a $b$-coloring $\coloring_{H}$ with $n+3$ colors. Since by Claim~\ref{claim:twins} the set $S$ contains $b$-vertices for at most three colors, we have that the vertices $v_1,\ldots,v_n$ are $b$-vertices in this coloring. Moreover, since they induce a clique in $H$, they all have distinct colors. Assume, without loss of generality, that $\coloring_{H}(\{v_1,\ldots,v_n\})=\{4,\ldots,n+3\}$. 
This implies that for each $i$, $\coloring_{H}(N(v_i)\cap\{u_1,\ldots,u_{\card{E(G)}}\})=\{1,2,3\}$. 
It follows that $\coloring_E \colon E(G) \to \bN$, defined as $\coloring_E(e_i) = \coloring_H(u_i)$, for $i \in \{1, \ldots, \card{E(G)}\}$, is a $3$-edge coloring of $G$. We argue that $\coloring_E$ is proper. Suppose for a contradiction that there exist adjacent edge $e_i$ and $e_j$, sharing the endpoint $v_s$, such that $\coloring_E(e_i) = \coloring_E(e_j) = c$. Since $\deg_G(v_s) = 3$, and two of its incident edges received the same color $c$, we can conclude that at least one of the colors $\{1, 2, 3\}$ does not appear in the neighborhood of $v_s$ in $H$, a contradiction with the fact that $v_s$ is a $b$-vertex of its color in $\coloring_H$.
\end{proof}

The previous theorem, together with the result that \bcol{} is \NP{}-complete when~$k = \Delta(G) + 1$ \cite{KTV02} and when~$k = m(G)$ \cite{HSS12}, proves Theorem~\ref{thm:main}(\ref{thm:main:np-c:01}). We now turn to the proof of Theorem~\ref{thm:main}(\ref{thm:main:np-c:2+}), that is, we show that \bcol{} remains \NP{}-complete for $k = 3$ if $k=\Delta(G)+1-p$ or $k=m(G)-p$ for any $p\geq 2$, based on a reduction due to Sampaio~\cite{Sam12}.

\begin{proposition}\label{prop:NP:c:fixed}
		For any fixed integer $c \ge 1$, it is $\NP$-complete to decide whether a graph~$G$ with $\Delta(G) = c + 3$ or $m(G) = c + 4$ has a $b$-coloring with $3$ colors.
\end{proposition}
\begin{proof}
Sampaio showed that the problem of deciding whether a graph $G$ has a $b$-coloring with $k$ colors is $\NP$-complete for any fixed $k \in \bN$~\cite[Proposition 4.5.1]{Sam12}. For the case of $k=3$, the reduction is from \textsc{$3$-Coloring} on planar $4$-regular graphs which is known to be $\NP$-complete~\cite{GJS76}. In this reduction, one takes the graph of the \textsc{$3$-Coloring} instance and adds three stars with two leaves each to the graph which can serve as the $b$-vertices in the resulting instance of \bcol{}.
Since this does not increase the maximum degree, we immediately have that the problem of deciding whether a graph of maximum degree four has a $b$-coloring with three colors is $\NP$-complete. Furthermore, by adding more leaves to one of the stars and thereby increasing the maximum degree of the graph in the resulting instance, we have that for any fixed integer $c \ge 1$, it is \NP{}-complete to decide whether a graph of maximum degree $\Delta(G) = c + 3$ has a $b$-coloring with three colors. 

Towards the statement regarding $m(G)$, we first observe that for a $4$-regular graph $G$ on at least five vertices, we have that $m(G) = 5$.
We observe that in any star with \emph{at least} two leaves, the center vertex can be a $b$-vertex in a coloring with three colors. We construct a graph $G'$ by adding five stars with four leaves each to $G$, and we again have that $G$ has a $3$-coloring if and only if $G'$ has a $b$-coloring with three colors, showing that the problem of deciding whether a graph $H$ with $m(H) = 5$ has a $b$-coloring with three colors, is \NP{}-complete. Note that in this reduction, the center vertices of the stars can be regarded as the vertices determining the $m$-degree of the graph in the resulting instance of $b$-coloring with three colors, so we can extend this result in a similar way as above. 
That is, for any $c \ge 1$, given a $4$-regular graph $G$, we can add $c + 4$ stars with $c + 4 - 1$ leaves each to $G$, implying that for the resulting graph $G'$, $m(G') = c + 4$. Again, $G$ has a $3$-coloring if and only if $G'$ has a $b$-coloring with three colors, implying the second statement of the proposition.
\end{proof}

We conclude this section by considering the complexity of the two problems on graphs with few vertices of high degree. Since \bprob{} and \bcol~ are known to be \NP{}-complete when $k=\Delta(G)+1$~\cite{KTV02}, we make the following observation which is of relevance to us since in Section~\ref{sec:algorithms:ell+k}, we show that \bcol{} is \FPT{} parameterized by $\Delta(G)+\ell_k(G)$.

\begin{observation} \label{obs:lkG}
\bprob{} and \bcol{} are \NP{}-complete on graphs with $\ell_k(G)=0$, where $k$ is the integer associated with the respective problem.
\end{observation}

\section{Dichotomy Algorithms}\label{sec:alg:dich}
In this section we give the algorithms in our dichotomy result, proving Theorem~\ref{thm:main}(\ref{thm:main:poly}). We show that for fixed $k \in \bN$, the problem of deciding whether a graph $G$ admits a $b$-coloring with $k$ colors is polynomial-time solvable when $k = m(G)$ (Section~\ref{sec:alg:dich:mG}), when $k = \Delta(G)$ (Section~\ref{sec:alg:dich:Delta}), and when $k = m(G) - 1$ (Section~\ref{sec:alg:dich:mG-1}), by providing \XP{}-algorithms for each case. Before we give the algorithms, we introduce the notion of a \emph{$b$-precoloring} and show how to enumerate all minimal $b$-precolorings of a graph.

\subsection{$b$-Precolorings}\label{sec:b-precol}
All algorithms in this section are based on guessing a proper coloring of several vertices in the graph, for which we now introduce the necessary terminology and establish some preliminary results.

\begin{definition}[Precoloring]
	Let $G$ be a graph and $k \in \bN$. A \emph{precoloring} with $k$ colors of a graph $G$ is an assignment of colors to a subset of its vertices, i.e.\ for $X \subseteq V(G)$, it is a map $\coloring_X \colon X \to [k]$. We call $\coloring_X$ \emph{proper}, if it is a proper coloring of $G[X]$.
	We say that a coloring $\coloring\colon V(G) \to [k]$ \emph{extends} $\coloring_X$, if $\coloring|_{X} = \coloring_X$.
\end{definition}

We use the following notation. For two precolorings $\coloring_X$ and $\coloring_Y$ with $X \cap Y = \emptyset$, we denote by $\coloring_X \cup \coloring_Y$ the precoloring that colors the vertices in $X$ according to $\coloring_X$ and the vertices in $Y$ according to $Y$, i.e.\ the precoloring $\coloring_{X \cup Y} \defeq \coloring_X \cup \coloring_Y$ defined as:
\begin{align*}
	\coloring_{X \cup Y}(v) = \left\lbrace\begin{array}{ll}\coloring_X(v), &\mbox{ if } v \in X \\ \coloring_Y(v), &\mbox{ if } v \in Y \end{array}\right. \mbox{ for all } v \in X \cup Y
\end{align*}

Next, we define a special type of precoloring with the property that any proper coloring that extends it is a $b$-coloring of the graph.

\begin{definition}[$b$-Precoloring]
	Let $G$ be a graph, $k \in \bN$, $X \subseteq V(G)$ and $\coloring_X$ a precoloring. We call $\coloring_X$ a \emph{$b$-precoloring with $k$ colors} if $\coloring_X$ is a $b$-coloring of $G[X]$. A $b$-precoloring $\coloring_X$ is called \emph{minimal} if for any $Y \subset X$, $\coloring_X|_{Y}$ is not a $b$-precoloring.
\end{definition}

It is immediate that any $b$-coloring can be obtained by extending a minimal $b$-precoloring, a fact that we capture in the following observation.

\begin{observation}\label{obs:bcol:bprecol}
	Let $G$ be a graph, $k \in \bN$, and $\coloring$ a $b$-coloring of $G$ with $k$ colors. Then, there is a set $X \subseteq V(G)$ such that $\coloring|_X$ is a minimal $b$-precoloring.
\end{observation}

The next observation captures the structure of minimal $b$-precolorings with $k$ colors. Roughly speaking, each such precoloring only colors  a set of $k$ $b$-vertices and for each $b$-vertex a set of $k-1$ of its neighbors that make that vertex the $b$-vertex of its color. We will use this property in the enumeration algorithm in this section to guarantee that we indeed enumerate all minimal $b$-precolorings with a given number of colors.

\begin{observation}\label{obs:min:precol:small}
	Let $\coloring_X$ be a minimal $b$-precoloring with $k$ colors. Then, $X = B \cup Z$, where 
	\begin{enumerate}[(i)]
		\item $B = \{x_1, \ldots, x_k\}$ and for $i \in [k]$, $\coloring_X(x_i) = i$, and
		\item $Z = \bigcup_{i \in [k]}{Z_i}$, where $Z_i \in \binom{N(x_i)}{k-1}$ and $\coloring_X(Z_i) = [k] \setminus \{i\}$.
	\end{enumerate}
\end{observation}

We are now ready to give the enumeration algorithm for minimal $b$-precolorings.

\begin{lemma}\label{lem:enum:precol}
	Let $G$ be a graph on $n$ vertices and $k \in \bN$. The number of minimal $b$-precolorings with $k$ colors of $G$ is at most
	\begin{align}\label{eq:beta}
		\beta(k) \defeq n^k\cdot \Delta^{k(k-1)} \cdot (k-1)!^k,
	\end{align}
	where~$\Delta \defeq \Delta(G)$ and they can be enumerated in time~$\beta(k)\cdot k^{\cO(1)}$.
\end{lemma}
\begin{proof}
	By Observation~\ref{obs:min:precol:small}, any minimal $b$-precoloring only colors a set of $k$ $b$-vertices, and for each of them a size-$(k-1)$ subset of its neighbors that are colored bijectively with the remaining colors.
	
	To guess all $b$-vertices in $G$, we enumerate all ordered vertex sets of size $k$, let $\{x_1, \ldots, x_k\}$ be such a set. Next, we enumerate all size-$(k-1)$ subsets of neighbors of each $x_i$ that can make $x_i$ the $b$-vertex of color $i$. Let $(Z_1, \ldots, Z_k)$ be a tuple of such sets of neighbors. Then we enumerate for each $i \in [k]$, all bijective colorings of $\pi_i \colon Z_i \to [k] \setminus \{i\}$ -- these are precisely the colorings of $Z_i$ that can make $x_i$ the $b$-vertex for color $i$. Given such a tuple $(\pi_1, \ldots, \pi_k)$, we make sure that it is consistent, i.e.\ for each vertex $v \in Z_i \cap Z_j$, we ensure that $\pi_i$ and $\pi_j$ assign $v$ the same color, i.e.\ that $\pi_i(v) = \pi_j(v)$. 
	If so, we construct a precoloring $\coloring_{B \cup Z}$ according to our choice of $B=\{x_1, \ldots, x_k\}$ and $(\pi_1, \ldots \pi_k)$ and if it is a minimal $b$-precoloring we output it. We give the details in Algorithm~\ref{alg:enum:precol}.
	\begin{algorithm}[h]
		\SetKwInOut{Input}{Input}\SetKwInOut{Output}{Output}
		\Input{A graph $G$, a positive integer $k$.}
		\Output{All minimal $b$-precolorings with $k$ colors of $G$}
		\ForEach{$B \in \binom{V(G)}{k}$ and every ordering $x_1, \ldots, x_k$ of the elements of $B$\label{alg:enum:precol:guess:B}}{
			\ForEach{$(Z_1, \ldots, Z_k) \in \binom{N(x_1)}{k-1} \times \cdots \times \binom{N(x_k)}{k-1}$\label{alg:enum:precol:guess:Z}}{
				\ForEach{$(\pi_1, \ldots, \pi_k)$,\label{alg:enum:precol:guess:pi} 
					\begin{enumerate}[(i)]
						\item where for all $i \in [k]$, $\pi_i \colon Z_i \to [k] \setminus \{i\}$ is a bijection, and\label{alg:enum:precol:guess:pi:1}
						\item for all $v \in Z_i \cap Z_j$, $\pi_i(v) = \pi_j(v)$\label{alg:enum:precol:guess:pi:2}
					\end{enumerate}}{
						Let $Z \defeq \cup_{i \in [k]} Z_i$ and $\coloring_{B \cup Z} \colon B \cup Z \to [k]$\;
						\lFor{$i \in [k]$}{$\coloring_{B \cup Z}(x_i) \defeq i$}
						\lFor{$i \in [k]$ and $v \in Z_i$}{$\coloring_{B \cup Z}(v) \defeq \pi_i(v)$}
						\lIf{$\coloring_{B \cup Z}$ is a minimal $b$-precoloring\label{alg:enum:precol:check}}{\textbf{output} $\coloring_{B \cup Z}$\label{alg:enum:precol:output} \textbf{and continue}}
					}
			}
		}
		\caption{Enumerating all minimal $b$-precolorings with $k$ colors of a graph.}
		\label{alg:enum:precol}
	\end{algorithm}
	
	We now show that the algorithm is correct.
	\begin{nestedclaim}
		A precoloring $\coloring_X$ is a minimal $b$-precoloring with $k$ colors if and only if Algorithm~\ref{alg:enum:precol} returns it in line~\ref{alg:enum:precol:output} in some iteration.
	\end{nestedclaim}
	\begin{claimproof}
		Suppose Algorithm~\ref{alg:enum:precol} returns a precoloring $\coloring_X$. We first argue that $\coloring_X$ is well-defined. Let $X = B \cup Z$ following the notation of Algorithm~\ref{alg:enum:precol}. It is immediate that $\coloring_X|_{B}$ is well-defined. For the remaining vertices $v \in Z$, we verify as condition (\ref{alg:enum:precol:guess:pi:2}) in line~\ref{alg:enum:precol:guess:pi} that whenever $v \in Z_i \cap Z_j$, $\pi_i(v) = \pi_j(v)$. Hence if the tuple $(\pi_1, \ldots, \pi_k)$ passes the check in line~\ref{alg:enum:precol:guess:pi}, then for each vertex $v \in Z$ there is precisely one value for $\coloring_X(v)$. We can conclude that $\coloring_X$ is well-defined. By the check performed in line~\ref{alg:enum:precol:check}, we can conclude that $\coloring_X$ is a minimal $b$-precoloring.
		
		Now suppose that $G$ contains a minimal $b$-precoloring $\coloring_X$. By Observation~\ref{obs:min:precol:small}, $X$ consists of an (ordered) set of $b$-vertices $B = \{x_1, \ldots, x_k\}$ with $\coloring_X(x_i) = i$ for $i \in [k]$, and a set $Z$ that contains, for each $x_i$, a set of $k-1$ neighbors $Z_i \subseteq Z$ such that $\coloring_X(Z_i) = [k] \setminus \{i\}$. Since Algorithm~\ref{alg:enum:precol} enumerates all such possible sets in lines~\ref{alg:enum:precol:guess:B} and~\ref{alg:enum:precol:guess:Z}, we have that in some iteration, it guessed $B \cup Z$ as the set of vertices to color. Since the algorithm enumerates all combinations of possibilities of coloring the sets $Z_i$ bijectively with colors $[k] \setminus \{i\}$ in line~\ref{alg:enum:precol:guess:pi}, it guessed a tuple of bijections $(\pi_1, \ldots, \pi_k)$ from which we obtain $\coloring_X$. Clearly we have that in that case, $(\pi_1, \ldots, \pi_k)$ passes the check in line~\ref{alg:enum:precol:guess:pi} and by assumption, $\coloring_X$ passes the check in line~\ref{alg:enum:precol:check}.
	\end{claimproof}
	
	It remains to argue its runtime. In line~\ref{alg:enum:precol:guess:B}, there are $\binom{n}{k}$ choices for the set $B$ and $k!$ choices for its orderings, in line~\ref{alg:enum:precol:guess:Z}, there are at most $\binom{\Delta}{k-1}^k$ choices of $k$-tuples of size-$(k-1)$ sets of neighbors, and in line~\ref{alg:enum:precol:guess:pi}, we enumerate $(k-1)!^k$ $k$-tuples of bijections of sets of size $k-1$. The remaining steps can be executed in time $k^{\cO(1)}$: 
	By construction, $\card{B \cup Z} \le k^2$, and every color has a $b$-vertex. It remains to verify whether the coloring $\coloring_{B \cup Z}$ is proper on $G[B \cup Z]$ to conclude that it is a $b$-precoloring. If so, we can verify minimality in polynomial time by simply trying for each vertex $x \in B \cup Z$, whether $\coloring_{B \cup Z \setminus \{x\}}$ is still a $b$-precoloring. If we can find such a vertex $x$, then $\coloring_{B \cup Z}$ is not minimal, otherwise it is.
	The total runtime amounts to
	\begin{align*}
		\binom{n}{k}k! \cdot \binom{\Delta}{k-1}^k \cdot (k-1)!^k \cdot k^{\cO(1)} \le n^k \cdot \Delta^{k(k-1)} \cdot (k-1)!^k \cdot k^{\cO(1)} = \beta(k) \cdot k^{\cO(1)},
	\end{align*}
	as claimed. The upper bound of $\beta(k)$ on the number of $b$-precolorings with $k$ colors follows since the $k^{\cO(1)}$ factor in the runtime only concerns the construction of the precolorings and the verification of whether they are indeed $b$-precolorings.
\end{proof}

\subsection{Algorithm for $k = m(G)$}\label{sec:alg:dich:mG}
Our first application of Lemma~\ref{lem:enum:precol} is to solve the \bcol{} problem in the case when $k = m(G)$ in time \XP{} parameterized by $k$. It turns out that in this case, we are dealing with a \yes{}-instance as soon as we found a $b$-precoloring in the input graph that also colors all high-degree vertices (see Claim~\ref{claim:xp:mG:cor}).
\begin{theorem}\label{thm:xp:mG}
	Let $G$ be a graph. There is an algorithm that decides whether $G$ has a $b$-coloring with $k = m(G)$ colors in time $n^{k^2} \cdot 2^{\cO(k^2 \log k)}$.
\end{theorem}
\begin{proof}
	Let $D \subseteq V(G)$ denote the set of vertices in $G$ that have degree at least $k$. Note that by the definition of $m(G)$, we have that $\card{D} \le k$.
	\begin{nestedclaim}\label{claim:xp:mG:cor}
		$G$ has a $b$-coloring with $k$ colors if and only if $G$ has a $b$-precoloring $\coloring_X$ such that $D \subseteq X$ and there exists $S\subseteq D$ such that $\coloring_X|_{(X \setminus S)}$ is a minimal $b$-precoloring.
	\end{nestedclaim}
	\begin{claimproof}
		Suppose $G$ has a $b$-precoloring $\coloring_X$ satisfying the condition of the claim. By our choice of $D$, each vertex in $V(G) \setminus D$ has degree at most $k-1$. Hence we can greedily compute and extension $\coloring$ of $\coloring_X$ that is a proper coloring of $G$. By the definition of~$b$-precoloring, we have that $\coloring$ is a $b$-coloring of $G$.
		
		Now suppose that $G$ has a $b$-coloring $\coloring$ with $k$ colors. Let $B = \{x_1, \ldots, x_k\}$ be the set of $b$-vertices of $\coloring$ and for each $i \in [k]$, let $Z_i$ be a set of $k-1$ neighbors of $x_i$ such that $\coloring(Z_i) = [k] \setminus \{i\}$. Let $Z \defeq \cup_{i \in [k]} Z_i$. Then, $\coloring|_{B \cup Z}$ is a $b$-precoloring. Clearly, $\coloring|_{B \cup Z}$ contains a minimal $b$-precoloring on vertex set $W\subseteq B\cup Z$. Then, $\coloring|_{W \cup D}$ is a $b$-precoloring of $G$ that satisfies the condition of the claim.
	\end{claimproof}
	
	The algorithm enumerates all minimal $b$-precolorings with $k$ colors and for each such precoloring, it enumerates all colorings of the vertices $D$. If combining one such pair of precolorings gives a $b$-precoloring, it returns a greedy extension of it; otherwise it reports that there is no $b$-coloring with $k$ colors, see Algorithm~\ref{alg:xp:mG:precol}.
	\begin{algorithm}[h]
		\SetKwInOut{Input}{Input}\SetKwInOut{Output}{Output}
		\Input{A graph $G$}
		\Output{A $b$-coloring with $m(G)$ colors if it exists, \no{} otherwise.}
		\ForEach{minimal $b$-precoloring $\coloring_X \colon X \to [k]$\label{alg:xp:mG:precol:enum:X}}{
			\ForEach{precoloring $\coloring_{D \setminus X} \colon (D \setminus X) \to [k]$\label{alg:xp:mG:precol:enum:D}}{
				\lIf{$\coloring_{X \cup D} \defeq \coloring_X \cup \coloring_{D \setminus X}$ is proper}{\Return a greedy extension of $\coloring_{X \cup D}$\label{alg:xp:mG:precol:yes}}
			}
		}
		\Return \no{}\;
		\caption{Algorithm for \bcol{} with $k = m(G)$.}
		\label{alg:xp:mG:precol}
	\end{algorithm}
	
	The correctness of the algorithm follows from the fact that it enumerates all precolorings that can satisfy Claim~\ref{claim:xp:mG:cor}. 
	We discuss its runtime. By Lemma~\ref{lem:enum:precol}, we can enumerate all minimal $b$-precolorings with $k$ colors in time $\beta(k) \cdot k^{\cO(1)}$. For each such minimal $b$-precoloring, we also enumerate all colorings of $D$. Since $\card{D} \le k$, this gives an additional factor of $k^k$ to the runtime which (with $\Delta \le n$) then amounts to
	\begin{align*}
		\beta(k) \cdot k^k \cdot k^{\cO(1)} = n^k \cdot \Delta^{k(k-1)} \cdot (k-1)!^k \cdot k^k \cdot k^{\cO(1)} \le n^{k^2} \cdot k!^k \cdot k^{\cO(1)} = n^{k^2} \cdot 2^{\cO(k^2\log k)},
	\end{align*}
	as claimed.
\end{proof}

\subsection{Algorithm for $k = \Delta(G)$}\label{sec:alg:dich:Delta}
Next, we turn to the case when $k = \Delta(G)$. Here the strategy is to again enumerate all minimal $b$-precolorings, and then for each such precoloring we check whether it can be extended to the remainder of the graph. Formally, we use an algorithm for the following problem as a subroutine.

\problemdef
	{Precoloring Extension (PrExt)}
	{A graph $G$, an integer $k$, and a precoloring $\coloring_X \colon X \to [k]$ of a set $X \subseteq V(G)$}
	{Does $G$ have a proper coloring with $k$ colors extending $\coloring_X$?}

Naturally, \textsc{Precoloring Extension} is a hard problem, since it includes \textsc{Graph Coloring} as the special case when $X = \emptyset$. However, when $\Delta(G) \le k - 1$, then the problem is trivially solvable: we simply check if the precoloring at the input is proper and if so, we compute an extension of it greedily. Since each vertex has degree at most $k-1$, there is always at least one color available. 
The case when $k = \Delta(G)$ has also been shown to be solvable in polynomial time.
\begin{theorem}[Thm.~3 in~\cite{CC06}, see also~\cite{DDJP15}]\label{thm:precoloring:extension}
	There is an algorithm that solves \textsc{Precoloring Extension} in polynomial time whenever $\Delta(G) \le k$.
\end{theorem}

\begin{theorem}\label{thm:xp:Delta}
	There is an algorithm that decides whether a graph $G$ has a $b$-coloring with $\Delta(G)$ colors in time 
	$n^{k + \cO(1)} \cdot 2^{\cO(k^2\log k)}$.
\end{theorem}
\begin{proof}
	The algorithm simply enumerates all minimal $b$-precolorings and then applies the algorithm for \textsc{PrExt} of Theorem~\ref{thm:precoloring:extension}. This algorithm can be applied with any precoloring of $G$ since $k = \Delta(G)$. We give the details in Algorithm~\ref{alg:xp:Delta:precol}.
	\begin{algorithm}
		\SetKwInOut{Input}{Input}\SetKwInOut{Output}{Output}
		\Input{A graph $G$}
		\Output{A $b$-coloring of $G$ with $k = \Delta(G)$ colors if it exists, and \no{} otherwise.}
		\ForEach{minimal $b$-precoloring $\coloring_X$ of $G$\label{alg:xp:Delta:precol:enum}}{
			Apply the algorithm for \textsc{PrExt} of Theorem~\ref{thm:precoloring:extension} with input $(G, k, \coloring_X)$\;
			\lIf{the algorithm found a proper coloring $\coloring$ extending $\coloring_X$}{\Return $\coloring$}
		}
		\Return \no{}\;
		\caption{Algorithm for \bcol{} with $k = \Delta(G)$.}
		\label{alg:xp:Delta:precol}
	\end{algorithm}
	
	We now show that the algorithm is correct.
	\begin{nestedclaim}\label{claim:alg:xp:Delta:correct}
		A graph $G$ contains a $b$-coloring with $k = \Delta(G)$ colors if and only if Algorithm~\ref{alg:xp:Delta:precol} returns a coloring $\coloring$.
	\end{nestedclaim}
	\begin{claimproof}
		Suppose Algorithm~\ref{alg:xp:Delta:precol} returns a coloring $\coloring$. Since $\coloring$ extends a $b$-precoloring $\coloring_X$ with $k$ colors, we can conclude that $\coloring$ has a $b$-vertex for each color. By the correctness of the algorithm of Theorem~\ref{thm:precoloring:extension}, we can conclude that $\coloring$ is a $b$-coloring with $k$ colors.
		
		Suppose $G$ contains a $b$-coloring with $k$ colors, say $\coloring$. By Observation~\ref{obs:bcol:bprecol}, $\coloring$ contains a minimal $b$-precoloring, say $\coloring_X$. Hence, Algorithm~\ref{alg:xp:Delta:precol:enum}, guessed $\coloring_X$ in some iteration. Furthermore, since $\coloring$ is a proper coloring that extends $\coloring_X$, $(G, k, \coloring_X)$ is a \yes{}-instance of \textsc{PrExt}, so the algorithm of Theorem~\ref{thm:precoloring:extension} returned a $b$-coloring $\coloring'$ that extends $\coloring_X$.
	\end{claimproof}
	It remains to argue the runtime. By Lemma~\ref{lem:enum:precol}, we can enumerate all $b$-precolorings in $\beta(k)\cdot k^{\cO(1)}$ time and  by Theorem~\ref{thm:precoloring:extension}, the algorithm for \textsc{PrExt} runs in time $n^{\cO(1)}$. The total runtime is hence (with $\Delta(G) = k$)
	\begin{align*}
		\beta(k)\cdot k^{\cO(1)} \cdot n^{\cO(1)} &= n^k \cdot \Delta^{k(k-1)} \cdot (k-1)!^k \cdot n^{\cO(1)} = n^{k + \cO(1)} \cdot k^{k(k-1)} \cdot (k-1)!^k \\ 
		&\le n^{k + \cO(1)} \cdot 2^{\cO(k^2\log k)},
	\end{align*}
	as claimed.
\end{proof}

\subsection{Algorithm for $k = m(G) - 1$}\label{sec:alg:dich:mG-1}
Before we proceed to describe the algorithm for \bcol{} when $k = m(G) - 1$, we show that the algorithm of Theorem~\ref{thm:precoloring:extension} can be used for a slightly more general case of \textsc{Precoloring Extension}, namely the case when all high-degree vertices in the input instance are precolored.
\begin{lemma}\label{lem:precoloring:extension}
	There is an algorithm that solves an instance $(G, k, \coloring_X)$ of \textsc{Precoloring Extension} in polynomial time whenever $\max_{v \in V(G) \setminus X} \deg(v) \le k$.
\end{lemma}
\begin{proof}
	First, we check whether $\coloring_X$ is a proper coloring of $G[X]$ and if not, the answer is \no{}. We create a new instance of \textsc{Precoloring Extension} $(G', k, \coloringg_{X'})$ as follows. For every vertex $x \in X$ and every vertex $y \in N_G(x) \setminus X$, we let $x_y$ be a new vertex that is only adjacent to $y$. We denote the set of these newly introduced vertices by $X' \defeq \{x_y \mid x \in X, y \in N_G(x) \setminus X\}$.
	We obtain $G'$ from $G$ as follows. Let $G'' = G - X$. Then, the vertex set of $G'$ is $V(G') \defeq V(G'') \cup X'$ and its edge set is $E(G') \defeq E(G'') \cup \{x_y y \mid x_y \in X'\}$.
	Now, we define a precoloring $\coloringg_{X'} \colon X' \to [k]$ such that for $x_y \in X'$, $\coloringg_{X'}(x_y) \defeq \coloring_X(x)$.
	
	It is clear that $(G, k, \coloring_X)$ is a \yes{}-instance of \textsc{Precoloring Extension} if and only if $(G', k, \coloringg_{X'})$ is a \yes{}-instance of \textsc{Precoloring Extension}.
	Furthermore, for every vertex $z \in X'$, $\deg_{G'}(z) = 1$ and for every vertex $v \in V(G') \setminus X'$, $\deg_{G'}(v) = \deg_G(v) \le k$, so $\Delta(G') \le k$. This means that we can solve the instance $(G', k, \coloring_{X'})$ in polynomial time using Theorem~\ref{thm:precoloring:extension}.
\end{proof}

\begin{theorem}\label{thm:xp:mG-1}
	There is an algorithm that decides whether a graph $G$ has a $b$-coloring with $k = m(G) - 1$ colors in time $n^{k^2 + \cO(1)} \cdot 2^{k^2 \log k}$.
\end{theorem}
\begin{proof}
	Let $D$ denote the set of vertices of degree at least $k+1$ in $G$. By the definition of $m(G)$, we have that $\card{D} \le k + 1$. 
	We first enumerate all minimal $b$-precolorings of $G$, for each such precoloring, we enumerate all precolorings of $D$. Then, given a $b$-precoloring $\coloring_X$ with $D \subseteq X$, we have that every vertex in $V(G) \setminus X$ has degree at most $k$, so we can apply the algorithm of Lemma~\ref{lem:precoloring:extension} to verify whether there is a proper coloring of $G$ that extends $\coloring_X$. If so, we output that extension. If no such precoloring can be found, then we conclude that we are dealing with a \no{}-instance. We give the details in Algorithm~\ref{alg:xp:mG-1:precol}.
	\begin{algorithm}[h]
		\SetKwInOut{Input}{Input}\SetKwInOut{Output}{Output}
		\Input{A graph $G$}
		\Output{A $b$-coloring of $G$ with $k = m(G) - 1$ colors if it exists, and \no{} otherwise.}
		\ForEach{minimal $b$-precoloring $\coloring_X$ of $G$\label{alg:xp:mG-1:precol:enum}}{
			\ForEach{precoloring $\coloring_{D \setminus X} \colon (D \setminus X) \to [k]$\label{alg:xp:mG-1:precol:enum:D}}{
				\If{$\coloring_{X \cup D} \defeq \coloring_X \cup \coloring_{D \setminus X}$ is proper}{
					Apply the algorithm for \textsc{PrExt} of Lemma~\ref{lem:precoloring:extension} with input $(G, k, \coloring_X)$\;
					\lIf{the algorithm found a proper coloring $\coloring$ extending $\coloring_{X \cup D}$}{\Return $\coloring$}
				}
			}
		}
		\Return \no{}\;			
		\caption{Algorithm for \bcol{} with $k = m(G) - 1$.}
		\label{alg:xp:mG-1:precol}
	\end{algorithm}
	
	We now prove the correctness of the algorithm.
	\begin{nestedclaim}
	$G$ has a $b$-coloring with $k = m(G) - 1$ colors if and only if Algorithm~\ref{alg:xp:mG-1:precol} returns a coloring $\coloring$.
	\end{nestedclaim}
	\begin{claimproof}
		Suppose Algorithm~\ref{alg:xp:mG-1:precol} returns a coloring $\coloring$. Then, $\coloring$ is obtained from a minimal $b$-precoloring $\coloring_X$ and a precoloring $\coloring_{D \setminus X}$, both with $k$ colors, such that $\coloring_{X \cup D} = \coloring_X \cup \coloring_{D \setminus X}$ is proper. Furthermore, since all vertices in $V(G) \setminus D$ have degree at most $k$, the application of the algorithm of Lemma~\ref{lem:precoloring:extension} returns a correct answer. Hence, $\coloring$ is a proper coloring and since it is obtained by extending a $b$-precoloring, it is a $b$-coloring with $k$ colors.
	
	The forward direction can be proved as in Claim~\ref{claim:alg:xp:Delta:correct} using Observation~\ref{obs:bcol:bprecol} which states that every $b$-coloring can be obtained by extending a minimal $b$-precoloring.
	\end{claimproof}
	It remains to argue the runtime. In line~\ref{alg:xp:mG-1:precol:enum}, we enumerate $\beta(k)$ (see (\ref{eq:beta})) minimal $b$-precolorings in time $\beta(k) \cdot k^{\cO(1)}$ using Lemma~\ref{lem:enum:precol}. For each such precoloring, we enumerate all precolorings of $D \setminus X$. Since $\card{D} \le k + 1$, there are at most $k^{k+1}$ such colorings. Finally, we run the algorithm for \textsc{PrExt} due to Lemma~\ref{lem:precoloring:extension} which takes time $n^{\cO(1)}$. The total runtime becomes
	\begin{align*}
		\beta(k) \cdot k^{\cO(1)} \cdot k^{k+1} \cdot n^{\cO(1)} = n^k \cdot \Delta^{k(k-1)} \cdot (k-1)!^k \cdot k^{k+1} \cdot n^{\cO(1)} \le n^{k^2 + \cO(1)} \cdot 2^{k^2 \log k},
	\end{align*}
	as claimed.
\end{proof}

\section{Maximum Degree Parameterizations}\label{sec:alg:max:deg}
In this section we consider parameterizations of \bcol{} that involve the maximum degree $\Delta(G)$ of the input graph $G$. In Section~\ref{sec:algorithms:mG} we show that we can solve \bcol{} when $k = m(G)$ in time \FPT{} parameterized by $\Delta(G)$ and in Section~\ref{sec:algorithms:ell+k} we show that \bcol{} is \FPT{} parameterized by $\Delta(G) + \ell_k(G)$. 

Both algorithms presented in this section make use of the following reduction rule, which has already been applied in~\cite{PPS17,Sam12} to obtain the \FPT{} algorithm for the problem of deciding whether a graph $G$ has a $b$-coloring with $k = \Delta(G) + 1$ colors, parameterized by $k$.
\begin{reductionrule}[\cite{PPS17,Sam12}]\label{rule:degree:k-2}
	Let $(G, k)$ be an instance of \bcol{}. If there is a vertex $v \in V(G)$ such that every vertex $w \in N[v]$ has degree at most $k-2$, then reduce $(G, k)$ to $(G - v, k)$.
\end{reductionrule}

\subsection{\FPT{} Algorithm for $k = m(G)$ parameterized by $\Delta(G)$}\label{sec:algorithms:mG}
Sampaio~\cite{Sam12} and Panolan et al.~\cite{PPS17} independently showed that parameterized by $\Delta(G)$, it can be decided in \FPT{} time whether a graph $G$ has a $b$-coloring with $\Delta(G) + 1$ colors. In this section we show that in the same parameterization, it can be decided in \FPT{} time whether a graph has a $b$-coloring with $m(G)$ colors.

\begin{theorem}\label{thm:mG:fpt}
	There is an algorithm that given a graph $G$ on $n$ vertices decides whether $G$ has a $b$-coloring with $k = m(G)$ colors in time $2^{\cO(k^4\cdot \Delta)} + n^{\cO(1)} < 2^{\cO(\Delta^5)} + n^{\cO(1)}$, where $\Delta \defeq \Delta(G)$.
\end{theorem}
\begin{proof}
	We apply Reduction Rule~\ref{rule:degree:k-2} exhaustively to $G$ and consider the following $3$-partition $(D, T, R)$ of $V(G)$, where $D$ contains the vertices of degree at least $k$, $T$ the vertices of degree precisely $k-1$ and $R$ the remaining vertices, i.e.\ $R \defeq V(G) \setminus (D \cup T)$.
	Since we applied Reduction Rule~\ref{rule:degree:k-2} exhaustively, we make
\begin{nestedobservation}\label{obs:rest:neighbor:densetight}
	Every vertex in $R$ has at least one neighbor in $D \cup T$.
\end{nestedobservation}
	We pick an inclusion-wise maximal set $B \subseteq D \cup T$ such that for each pair of distinct vertices $b_1, b_2 \in B$, we have that $\dist(b_1, b_2) \ge 4$.

\smallskip
\noindent\textbf{Case 1 ($\card{B \cap T} < k$).}\footnote{This case is almost identical to~\cite[Case II in the proof of Theorem 2]{PPS17}.} We show that for any vertex in $u \in V(G) \setminus B$, there is a vertex $v \in B$ such that $\dist(u, v) \le 4$. Suppose $u \in D \cup T$. Since we did not include $u$ in $B$, it immediately follows that there is some $v \in B$ such that $\dist(u, v) < 4$. Now suppose $u \in R$. By Observation~\ref{obs:rest:neighbor:densetight}, $u$ has a neighbor $w$ in $D \cup T$ and by the previous argument, there is a vertex $v \in B$ such that $\dist(w, v) < 4$. We conclude that $\dist(u, v) \le 4$. Using this observation, we now show that in this case, the number of vertices in $G$ is polynomial in $k$ and $\Delta$.
\begin{nestedclaim}\label{claim:case:1:size}
	If $\card{B \cap T} < k$, then $\card{V(G)} \le \cO(k^4 \cdot \Delta)$.
\end{nestedclaim}
\begin{claimproof}
Note that $(B\cup D, S_1,\ldots, S_4)$ constitutes a partition of $V(G)$, where $S_i$ is the set of vertices of $V(G)\setminus (B\cup D)$ that are at distance exactly $i$ from $B$. Since $\card{B \cap T} < k$ and $\card{D} \le k$, we have that $\card{B\cup D} < 2k$, and therefore $\card{S_1}<2k\cdot\Delta$. By the definition of $m(G)$, all the vertices in $S_1\cup\ldots\cup S_4$ have degree at most $k-1$. This implies that $\card{S_i}<(k-1)^{i-1}\cdot 2k\cdot\Delta$. We conclude that the number of vertices in $G$ is at most $2k+ 2k\cdot\Delta\cdot\sum^{4}_{i=1}(k-1)^{i-1}=\cO(k^4\cdot\Delta)$.
\end{claimproof} 

By Claim~\ref{claim:case:1:size}, we can solve the instance in Case 1 in time $2^{\cO(k^4 \cdot \Delta)}$ using the algorithm of Panolan et al.~\cite{PPS17}.

\smallskip
\noindent\textbf{Case 2 ($\card{B \cap T} \ge k$).} Let $B' \subseteq B \cap T$ with $\card{B'} = k$ and denote this set by $B' = \{x_1, x_2, \ldots, x_k\}$. We show that we can construct a $b$-coloring $\coloring \colon V(G) \to [k]$ of $G$ such that for $i \in [k]$, $x_i$ is the $b$-vertex of color $i$. For $i \in [k]$, we let $\coloring(x_i) \defeq i$.
Next, we color the vertices in $D$. Recall that $\card{D} \le k$, so we can color the vertices in $D$ injectively with colors from $[k]$, ensuring that this will not create a conflict on any edge in $G[D]$. Furthermore, consider $i, j \in [k]$ with $i \neq j$. Since $\dist(x_i, x_j) \ge 4$, we have that $N(x_i) \cap N(x_j) = \emptyset$. In particular, there is no vertex in $D$ that has two or more neighbors in $B'$. To summarize, we can conclude that we can let $\coloring$ color the vertices of $D$ in such a way that:
\begin{enumerate}[(i)]
	\item $\coloring$ is injective on $D$, and
	\item $\coloring$ is a proper coloring of $G[B' \cup D]$.
\end{enumerate}
These two items imply that for each $x_i$ ($i \in [k]$), its neighbors $N(x_i) \cap D$ receive distinct colors which are also different from $i$. Let $\ell \defeq \card{N(x_i) \cap D}$. It follows that we can let $\coloring$ color the remaning $(k-1) - \ell$ neighbors of $x_i$ in an arbitrary bijective manner with the $(k-1) - \ell$ colors that do not yet appear in the neighborhood of $x_i$.

After this process, $x_i$ is a $b$-vertex for color $i$. We proceed in this way for all $i \in [k]$.
Since for $i, j \in [k]$ with $i \neq j$ we have that $\dist(x_i, x_j) \ge 4$, it follows that there are no edges between $N[x_i]$ and $N[x_j]$ in $G$.
Hence, we did not introduce any coloring conflict in the previous step. Now, all vertices in $G$ that have not yet received a color by $\coloring$ have degree at most $k-1$, so we can extend $\coloring$ to a proper coloring of $G$ in a greedy fashion. 

We summarize the whole procedure in Algorithm~\ref{alg:mG}.
	\begin{algorithm}	[t]
		\SetKwInOut{Input}{Input}\SetKwInOut{Output}{Output}
		\Input{A graph $G$ with $k = m(G)$~\tcp*[h]{More generally, graph $G$ with $\ell_k(G) \le k$}}
		\Output{A $b$-coloring with $k$ colors of $G$ if it exists, and \no{} otherwise.}
		
		Apply Reduction Rule~\ref{rule:degree:k-2} exhaustively\;\label{alg:mG:start}
		Let $(D, T, R)$ be a partition of $V(G)$ such that for all $x \in D$, $\deg_G(x) \ge k$, for all $x \in T$, $\deg_G(x) = k - 1$, and $R = V(G) \setminus (D \cup T)$\label{alg:mG:partition}\;
		Let $B \subseteq D \cup T$ be a maximal set such that for distinct $b_1, b_2 \in B$, $\dist(b_1, b_2) \ge 4$\label{alg:mG:dist}\;
		\If(\tcp*[h]{Case 1}){$\card{B \cap T} < k$\label{alg:mG:case:1:check}}{
			Solve the instance in time $2^{\cO(k^4 \cdot \Delta)}$ using~\cite{PPS17}\;\label{alg:mG:case:1}
			\If{the algorithm of \cite{PPS17} returned a $b$-coloring $\coloring$}
			{
				\Return $\coloring$\;
			} \Else {
				\Return \no{}\;
			}
		}\label{alg:mG:case:1:end}
		\Else(\tcp*[h]{Case 2, i.e.\ $\card{B \cap T} \ge k$}){
			Pick a size-$k$ subset of $B \cap T$, say $B' \defeq \{x_1, \ldots, x_k\}$\;			
			Initialize a $k$-coloring $\coloring \colon V(G) \to [k]$\;
			For $i \in [k]$, let $\coloring(x_i) \defeq i$\;
			Let $\coloring$ color the vertices of $D$ injectively such that $\coloring$ remains proper on $G[B' \cup D]$\;\label{alg:mG:d:injective}
			For $i \in [k]$, let $\coloring$ color $N(x_i) \cap D$ such that $x_i$ is the $b$-vertex of color $i$\;
			Extend the coloring $\coloring$ greedily to the remainder of $G$\;
			\Return $\coloring$\;
		}
		\caption{An $\FPT$-algorithm that decides whether a graph $G$ has a $b$-coloring with $m(G)$ colors, parameterized by $\Delta(G)$.}
		\label{alg:mG}
	\end{algorithm}
	We now analyze its runtime. Clearly, exhaustively applying Reduction Rule~\ref{rule:degree:k-2} can be done in time $n^{\cO(1)}$. As mentioned above, Case 1 can be solved in time $2^{\cO(k^4 \cdot \Delta)}$. In Case 2, the coloring of $G[B' \cup D]$ can be found in time $\cO(k^2)$, and extending the coloring to the remainder of $G$ can be done in time $n^{\cO(1)}$. The claimed bound follows.
\end{proof}
	We observe that Algorithm~\ref{alg:mG} in fact solves a more general case of the \bcol{} problem, a fact which we will use later in the proof of Theorem~\ref{thm:k+ell:fpt}. By the definition of $m(G)$, $\ell_{m(G)} \le m(G)$, and this is the only property of $m(G)$ that the algorithm relies on: it bounds the size of $D$ by $\card{D} \le m(G)$. This is crucially used in lines~\ref{alg:mG:case:1} and~\ref{alg:mG:d:injective}. 
		Now, if we relax the condition of $k = m(G)$ to $\ell_k(G) \le k$, we observe that the assumption $\ell_k(G) \le k$ still guarantees that $\card{D} \le k$. Hence, in line~\ref{alg:mG:case:1}, the bound of $\cO(k^4 \cdot \Delta)$ on $V(G)$ remains the same and in line~\ref{alg:mG:d:injective} we can find a coloring that is injective on $D$ as well.
	\begin{remark}\label{rem:mG:ell:k}
		Algorithm~\ref{alg:mG} solves the problem of deciding whether $G$ admits a $b$-coloring with $k$ colors in time $2^{\cO(k^4 \cdot \Delta)} + n^{\cO(1)}$ whenever $\ell_k(G) \le k$.
	\end{remark}
	Furthermore, in the proof of Theorem~\ref{thm:mG:fpt}, we in fact provide a polynomial kernel for the problem: In Case 1, we have a kernelized instance on $\cO(k^4 \cdot \Delta)$ vertices (see Claim~\ref{claim:case:1:size}) and in Case 2, we always have a \yes{}-instance.
	\begin{corollary}\label{cor:mG:kernel}
		The problem of deciding whether a graph $G$ has a $b$-coloring with $k = m(G)$ colors admits a kernel on $\cO(k^4 \cdot \Delta) = \cO(\Delta^5)$ vertices.
	\end{corollary}

\subsection{\FPT{} Algorithm Parameterized by $\Delta(G) + \ell_k(G)$}\label{sec:algorithms:ell+k}
The next parameterization of \bcol{} involving the maximum degree that we consider is by $\Delta(G) + \ell_k(G)$. We show that in this case, the problem is \FPT{}.
By Observation~\ref{obs:lkG} we know that \bcol{} is \NP{}-complete on graphs with $\ell_k(G)=0$, and by Theorem~\ref{thm:main}, it is \NP{}-complete even when $k = 3$ and $\Delta(G) = 4$. Hence, there is no \FPT{}- nor \XP{}-algorithm for a parameterization using only one of the two above mentioned parameters unless $\P = \NP$.
Note that the algorithm we provide in this section can be used to solve the case of $k = m(G)$ for which we gave a separate algorithm in Section~\ref{sec:algorithms:mG}, see Algorithm~\ref{alg:mG}. However, Algorithm~\ref{alg:mG} is much simpler than the algorithm presented in this section, and simply applying the following algorithm for the case $k = m(G)$ results in a runtime of $2^{\cO(k^{k + 2} \cdot \Delta)} + n^{\cO(1)}$ which is far worse than the runtime of $2^{\cO(k^4\cdot \Delta)} + n^{\cO(1)}$ of Theorem~\ref{thm:mG:fpt}.

\newcommand\eff{\cO(\ell\cdot\Delta\cdot\min\{\ell,\Delta\}^{\ell+1})}

\begin{theorem}\label{thm:k+ell:fpt}
	There is an algorithm that given a graph $G$ on $n$ vertices decides whether $G$ has a $b$-coloring with $k$ colors in time 
	$2^{\eff} + n^{\cO(1)}$, where $\Delta \defeq \Delta(G)$ and $\ell \defeq \ell_k(G)$.
\end{theorem}
\begin{proof}
	The overall strategy of the algorithm is similar to Algorithm~\ref{alg:mG}. We can make the following assumptions. First, if $\ell \le k$, then we can apply Algorithm~\ref{alg:mG} directly to solve the instance at hand, see Remark~\ref{rem:mG:ell:k}. Hence we can assume that $k < \ell$. Furthermore, $k \le \Delta + 1$, otherwise we are dealing with a trivial \no{}-instance; we have that $k \le \min\{\ell - 1, \Delta + 1\}$. Furthermore, we can assume that $k > 2$, otherwise the problem is trivially solvable in time polynomial in $n$.
	
	Again, we consider a partition $(D, T, R)$ of $V(G)$, where the vertices in $D$ have degree at least $k$, the vertices in $T$ have degree $k-1$ and the vertices in $R$ have degree less than $k - 1$. We assume that Reduction Rule~\ref{rule:degree:k-2} has been applied exhaustively, so Observation~\ref{obs:rest:neighbor:densetight} holds, i.e.\ every vertex in $R$ has at least one neighbor in $D \cup T$.
	
	Now, we pick an inclusion-wise maximal set $B \subseteq D \cup T$ such that for each pair of distinct vertices $b_1, b_2 \in B$, $\dist(b_1, b_2) \ge \ell + 2$.

	\smallskip
	\noindent\textbf{Case 1 ($\card{B \cap T} < k$).} By the same argument given in Case 1 of the proof of Theorem~\ref{thm:mG:fpt}, we have that any vertex in $T \cup R$ is at distance at most $\ell + 2$ from a vertex in $B$. We now give a bound on the number of vertices in $G$ in terms of $\ell$ and $\Delta$.
	\begin{nestedclaim}
		If $\card{B \cap T} < k$, then $\card{V(G)}=\cO(\ell\cdot\Delta\cdot\min\{\ell,\Delta\}^{\ell+1})$.
	\end{nestedclaim}
	\begin{claimproof}
The proof strategy is the same as in the proof of Claim~\ref{claim:case:1:size}. Note that $(B\cup D, S_1,\ldots, S_{\ell+2})$ constitutes a partition of $V(G)$, where $S_i$ is the set of vertices of $V(G)\setminus (B\cup D)$ that are at distance exactly $i$ from $B$. Since $\card{B \cap T} < k$ and $\card{D} \le \ell$, we have that $\card{B\cup D} < k+\ell$ and $\card{S_1}<\ell\cdot\Delta+k(k-1)=\cO(\ell\cdot\Delta)$. By the definition of the set $D$, all the vertices in $S_1\cup\ldots\cup S_{\ell+2}$ have degree at most $k-1$. Thus, $\card{S_i}=(k-1)\cdot\card{S_{i-1}} = \card{S_1} \cdot (k-1)^{i-1}$ for all $i \in \{2,\ldots, \ell+2\}$. We conclude that the number of vertices in $G$ is at most 
	\begin{align*}
		k + \ell + \card{S_1} \cdot \sum_{i = 1}^{\ell + 2} (k-1)^{i-1} = k + \ell + \card{S_1} \cdot \cO((k-1)^{\ell+1}) = \cO(\ell \cdot \Delta \cdot (k-1)^{\ell + 1}),
	\end{align*}
where $(k-1)\leq\min\{\ell-2,\Delta\} \le \min\{\ell, \Delta\}$ and therefore $\card{V(G)}=\eff$.
\end{claimproof}

	By the previous claim, we can solve the instance in time $2^{\eff}$ in this case, using the algorithm~\cite{PPS17}.
	
	\newcommand\nid{N_{i\mid D}}
	\smallskip
	\noindent\textbf{Case 2 ($\card{B \cap T} \ge k$).} Let $B' \subseteq B \cap T$ be of size $k$ and denote it by $B' \defeq \{x_1, \ldots, x_k\}$. The strategy in this case is as follows: We compute a proper coloring of $G[D]$, and then modify it so that can be extended to a $b$-coloring of $G$. In this process we will be able to guarantee for each $i \in [k]$, that either $x_i$ can be the $b$-vertex for color $i$, or we will have found another vertex in $D$ that can serve as the $b$-vertex of color $i$. The difficulty here arises from the following situation: Suppose that in the coloring we computed for $G[D]$, a vertex $x_i$ has two neighbors in $D$ that received the same color. Then, $x_i$ cannot be the $b$-vertex of color $i$ in any extension of that coloring, since $\deg(x_i) = k - 1$, and $k-1$ colors need to appear the neighborhood of $x_i$ for it to be a $b$-vertex. However, recoloring a vertex in $N(x_i) \cap D$ might create a conflict in the coloring of $G[D]$. These potential conflicts can only appear in the connected component of $G[D \cup B']$ that contains~$x_i$. We now show that each component of $G[D \cup B']$ can contain at most one such vertex, by our choice of the set $B$.
	\begin{nestedclaim}\label{claim:at:most:one:b:vertex}
		Let $C$ be a connected component of $G[D \cup B']$. Then, $C$ contains at most one vertex from $B'$.
	\end{nestedclaim}
	\begin{claimproof}
		Let $Z \defeq V(C) \cap B'$ and assume for the sake of a contradiction that $\card{Z} > 1$. Let $x_i, x_j \in Z$ be a pair of distinct vertices in $Z$ such that $\dist_{G[D \cup B']}(x_i, x_j)$ is minimized among all pairs of distinct vertices in $Z$. Hence, all vertices on the path from $x_i$ to $x_j$ in $G[D \cup B']$ are from $D$. Since $\card{D} = \ell$, we have that $\dist_{G[D \cup B']}(x_i, x_j) \le \ell + 1$. However, we then have that $\dist_G(x_i, x_j) \le \dist_{G[D \cup B']}(x_i, x_j) \le \ell + 1$, a contradiction with the choice of $B$, by which we have that $\dist_G(x_i, x_j) \ge \ell + 2$. 
	\end{claimproof}
	Throughout the following, for $i \in [k]$, we denote by $C_i$ the connected component of $G[D \cup B']$ that contains $x_i$, and by $\ell_i$ the number of vertices of $C_i$, i.e.\ $\ell_i \defeq \card{V(C_i)}$.
	By Claim~\ref{claim:at:most:one:b:vertex}, 
	\begin{align*}
		C_i \neq C_j, \mbox{ for all } i, j \in [k], i \neq j.
	\end{align*} 
	Let furthermore $\cC_\emptyset$ be the set of connected components of $G[D \cup B']$ that do not contain any vertex from $B'$. We observe that any proper coloring of $G[D \cup B']$ can be obtained from independently coloring the vertices in $C_1, \ldots, C_k$, and $\cC_\emptyset$. If for some $i \in [k]$, $C_i$ is a trivial\footnote{We call a connected component of a graph \emph{trivial} if it contains only one vertex.} component, then $N(x_i) \cap D = \emptyset$. Hence, we can assign $x_i$ any color without creating any conflict with the remaining vertices in $G[D \cup B']$.
	We illustrate the structure of $G$ in Figure~\ref{fig:alg:k+ell}.
	\begin{figure}
		\centering
		\includegraphics[height=.235\textheight]{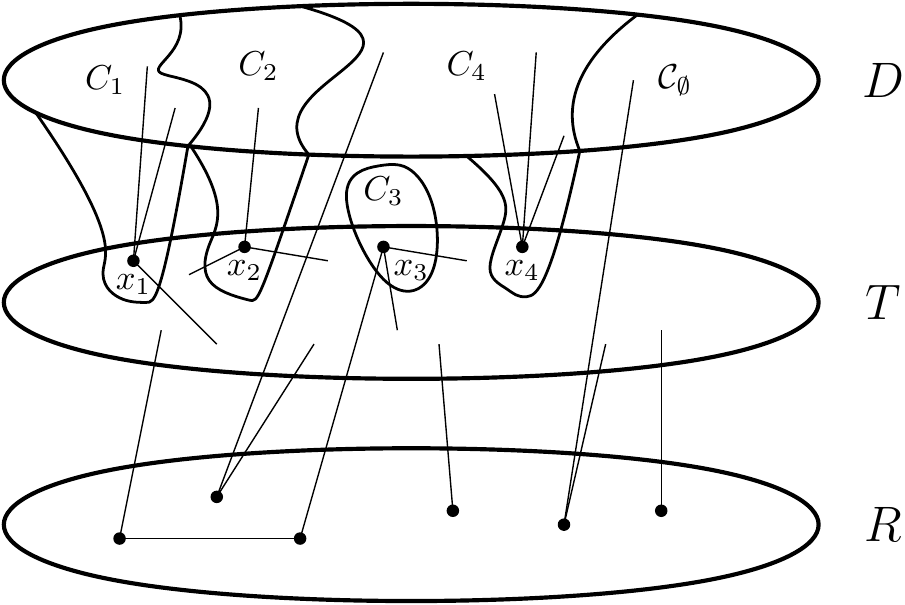}
		\caption{Illustration of the structure of a graph $G$ in the proof of Theorem~\ref{thm:k+ell:fpt} where $k = 4$. Here, $B' = \{x_1, \ldots, x_4\}$ and $C_1, \ldots, C_4$ are the components of $G[D \cup B']$ containing $x_1, \ldots, x_4$, respectively. Note that all vertices in $T$ are of degree $3$, all vertices in $R$ of degree at most $2$ and all vertices in $R$ have a neighbor in $D \cup T$.}
		\label{fig:alg:k+ell}
	\end{figure}
	Before we proceed with the proof of the next claim, we introduce some notation. For $X \subseteq V(G)$, a (pre-) coloring $\coloring \colon X \to [k]$, and $i, j \in [k]$, we denote by $\coloring_{i \leftrightarrow j}$ the (pre-) coloring obtained from $\coloring$ by \emph{switching} colors $i$ and $j$, i.e.\ for $v \in X$ we let:
	\begin{align*}
		\coloring_{i \leftrightarrow j}(v) \defeq \left\lbrace\begin{array}{ll} \coloring(v), &\mbox{ if } \coloring(v) \notin \{i, j\} \\
				i, &\mbox{ if } \coloring(v) = j \\
				j, &\mbox{ if } \coloring(v) = i
			\end{array}\right.
	\end{align*}
	It is immediate that $\coloring$ is proper if and only if $\coloring_{i \leftrightarrow j}$ is proper.
	\begin{nestedclaim}\label{claim:nontrivial:recoloring}
		Let $i \in [k]$ be such that $C_i$ is a nontrivial component of $G[D \cup B']$ and let $\coloring \colon V(C_i) \to [k]$ be a proper coloring of $C_i$. Then, one can obtain in time $\cO(k^2 \cdot \ell_i^2)$ a proper coloring $\coloringg \colon V(C_i) \to [k]$ of $C_i$ such that either
		\begin{enumerate}[(i)]
			\item there is a vertex in $V(C_i)$ different from $x_i$ that is a $b$-vertex for color $i$ in $\coloringg$, or\label{claim:nontrivial:recoloring:1}
			\item $\coloringg(x_i) = i$ and $\coloringg$ is injective on $N_{C_i}[x_i]$.\label{claim:nontrivial:recoloring:2} 
		\end{enumerate}
	\end{nestedclaim}
\begin{claimproof}
	We can assume that $\coloring(x_i) = i$, otherwise we let $\coloring \defeq \coloring_{i \leftrightarrow j}$ for some $[k] \ni j \neq i$.		
		If $\coloring$ is injective on $N_{C_i}[x_i]$, we let $\coloringg \defeq \coloring$ and we are in case (\ref{claim:nontrivial:recoloring:2}).
	
		Otherwise, we do as follows. Let $j \in [k]$ with $j \neq i$ be a color that does not appear on any vertex in $N_{C_i}(x_i)$, i.e.\ there is no vertex $y \in N_{C_i}(x_i)$ such that $\coloring(y) = j$. Such a color must exists by the fact that $\coloring$ is not injective on $N_{C_i}[x_i]$ and the fact that $\deg_G(x_i) = k-1$. For each vertex $z \in V(C_i)$ with $\coloring(z) = j$, we do the following.
		\begin{enumerate}[1)]
			\item If $\coloring(N[z]) = [k]$, i.e.\ if all colors except $j$ appear in the neighborhood of $z$, then $z$ is a $b$-vertex for color $j$. We let $\coloringg \defeq \coloring_{i \leftrightarrow j}$ and we are in case (\ref{claim:nontrivial:recoloring:1}).
			\item Otherwise, there is a color $j' \neq j$ that does not appear in the neighborhood of $z$. We update $\coloring$ by setting $\coloring(z) \defeq j'$, keeping the coloring $\coloring$ proper.
		\end{enumerate}
		If these two steps are executed for all vertices that $\coloring$ colored $j$ without ending up in case (\ref{claim:nontrivial:recoloring:1}), then $\coloring$ is a proper coloring of $C_i$ with colors $[k] \setminus \{j\}$. Now, let $y_1, y_2 \in N_{C_i}(x_i)$ be a pair of non-adjacent neighbors of $x_i$. We add the edge $y_1 y_2$ to $C_i$ and update $\coloring(y_1) \defeq j$. Now, $\coloring$ is a proper $k$-coloring of the graph obtained from $C_i$ by adding an edge in the neighborhood of $x_i$.
		
		We repeat this process until we either reached case (\ref{claim:nontrivial:recoloring:1}) at some stage, or we have that $\coloring$ is a proper $k$-coloring of the graph obtained from $C_i$ by making $N_{C_i}[x_i]$ a clique. The latter case implies that $\coloring$ is injective on $N_{C_i}[x_i]$ and we are in case (\ref{claim:nontrivial:recoloring:2}) by letting $\coloringg \defeq \coloring$.
		This recoloring procedure terminates within time $\cO(\binom{k-1}{2} \cdot \card{C_i}^2) = \cO(k^2\cdot \ell_i^2)$.
	\end{claimproof}
	The algorithm to solve this case now works as follows. First, we compute a proper $k$-coloring $\coloring$ of $G[D \cup B']$. We derive from $\coloring$ another $k$-coloring $\coloringg$ of $G[D \cup B']$. 
	For each $i \in [k]$, we do the following.
	If $C_i$ is a nontrivial component, then, with input $\coloring|_{V(C_i)}$
	we compute a proper $k$-coloring $\coloringg_i$ of $C_i$ using Claim~\ref{claim:nontrivial:recoloring} satisfying the stated conditions, and let $\coloringg|_{V(C_i)} \defeq \coloringg_i$.
	Finally, we let $\coloringg|_{V(\cC_\emptyset)} \defeq \coloring|_{V(\cC_\emptyset)}$.	
	
	We now show how to extend the coloring $\coloringg$ to a $b$-coloring of the entire graph $G$.	
	Let $i \in [k]$. By Claim~\ref{claim:nontrivial:recoloring} we know that in $\coloringg$ either there is a $b$-vertex for color $i$ in $C_i$ or $\coloringg$ is injective on $N_{C_i}[x_i]$. In the latter case, let $N_G(x_i) = \{y_1, \ldots, y_{k-1}\}$ and assume wlog.\ that for some $k' \le k-1$, $N_{C_i}(x_i) = \{y_1, \ldots, y_{k'}\}$. Then, $k'$ different colors appear on $\{y_1, \ldots, y_{k'}\}$ and we can let $\coloringg$ color $\{y_{k'+1}, \ldots, y_{k-1}\}$ bijectively with the remaining $k - 1 - k'$ colors to make $x_i$ the $b$-vertex of color $i$. (Note that since $\dist(x_i, x_j) \ge \ell + 2 \ge 4$ for $i \neq j$, this does not introduce any coloring conflict.) The remaining vertices of $G$ have degree at most $k-1$, so we can extend the coloring $\coloringg$ greedily to the remainder of $G$.
	
	It remains to argue the runtime of the algorithm. Applying Reduction Rule~\ref{rule:degree:k-2} exhaustively can be done in time $n^{\cO(1)}$. As  mentioned above, in Case 1 we can solve the instance in time $2^{\eff}$.
	In Case 2, we can compute a proper $k$-coloring of $G[D \cup B']$ in time $\cO^*(2^{\ell+k}) = 2^{\ell + k} \cdot \ell^{\cO(1)}$ using standard methods~\cite{BHK09}. Modifying this coloring to satisfy the condition of Claim~\ref{claim:nontrivial:recoloring} for each $i \in [k]$ takes time 
	$
		\cO(\sum_{i = 1}^k k^2 \cdot \ell_i^2) = \cO(k^3\cdot \sum_{i = 1}^k \ell_i^2) = \cO(k^3\cdot \ell^2) = \cO(\ell^5).
	$
	Extending the coloring to the remainder of $G$ can be done in time $n^{\cO(1)}$, so the total runtime of the algorithm is 
	\begin{align*}
		2^{\eff} + 2^{\ell + k} \cdot \ell^{\cO(1)} + \cO(\ell^5) + n^{\cO(1)} = 2^{\cO(\ell^{\ell + 2} \cdot \max\{\ell, \Delta\})} + n^{\cO(1)},
	\end{align*}	
	as claimed.
\end{proof}
	Similar to above, we have the following consequence.
	\begin{corollary}\label{cor:ell:kernel}
		The problem of deciding whether a graph $G$ admits a $b$-coloring with $k$ colors admits a kernel on $\eff$ vertices, where $\Delta \defeq \Delta(G)$ and $\ell \defeq \ell_k(G)$.
	\end{corollary}

\section{Conclusion}\label{sec:conclusion}

We have presented a complexity dichotomy for \bcol{} with respect to two upper bounds on the $b$-chromatic number, in the following sense: We have shown that given a graph $G$ and for fixed $k \in \{\Delta(G) + 1 - p, m(G) - p\}$, it can be decided in polynomial time whether $G$ has a $b$-coloring with $k$ colors whenever $p \in \{0, 1\}$ and the problem remains \NP{}-complete whenever $p \ge 2$, already for $k = 3$.

The most immediate question left open in this work is the parameterized complexity of the \bcol{} problem when $k \in \{m(G), \Delta(G), m(G) - 1\}$. In all of these cases, we have provided \XP{}-algorithms, and it would be interesting to see whether these problems are \FPT{} or \W{}[1]-hard.
\begin{openproblem}\label{prob:fpt}
	Let $G$ be a graph and $k \in \{m(G), \Delta(G), m(G) - 1\}$. Is the problem of deciding whether a graph $G$ has a $b$-coloring with $k$ colors parameterized by $k$ fixed-parameter tractable or \W{}[1]-hard?
\end{openproblem}
We showed that \bcol{} is \FPT{} parameterized by $\Delta(G) + \ell_k(G)$, where $\ell_k(G)$ denotes the number of vertices of degree at least $k$ in $G$, and this is optimal in the sense that there is no \FPT{} nor \XP{} algorithm for the problem parameterized by only one of the two invariants.  It would be interesting to see if one could devise an \FPT{}-algorithm for the parameterization that replaces the maximum degree by the number of colors.

\begin{openproblem}
	Is \bcol{} parameterized by $k + \ell_k(G)$ fixed-parameter tractable?
\end{openproblem}

Note that a positive answer to this question would also imply an \FPT{}-algorithm for the question of whether a graph $G$ has a $b$-coloring with $k=m(G)$ colors parameterized by $k$, partially answering Open Problem~\ref{prob:fpt}.

Recently, Effantin et al.~\cite{EGT16} introduced the \emph{relaxed $b$-chromatic number} of a graph $G$, $\chi_b^r(G)$, as the maximum $b$-chomatic number of any induced subgraph of $G$, i.e.~$\chi_b^r(G) \defeq \max_{X \subseteq V(G)} \chi_b(G[X])$. It is clear that $\chi_b(G) \le \chi_b^r(G)$, so it would be interesting to see if for fixed $k$, the problem of deciding whether a graph $G$ admits a $b$-coloring with $k$ colors when the value of $k$ is close to $\chi_b^r(G)$ admits a similar dichotomy as the ones we presented for the upper bounds $\Delta(G) + 1$ and $m(G)$ on $\chi_b(G)$. 

\paragraph*{Acknowledgements.} We would like to thank Fedor Fomin for useful advice with good timing and Petr Golovach for pointing out the reference~\cite{DDJP15}.

\newcommand{\noop}[1]{}
\bibliographystyle{siam}
\bibliography{ref}

\begin{thebibliography}{10}

\bibitem{Marx17}
{\sc P.~Aboulker, N.~Brettell, F.~Havet, D.~Marx, and N.~Trotignon}, {\em
  Coloring graphs with constraints on connectivity}, Journal of Graph Theory,
  85 (2017), pp.~814--838.

\bibitem{BCF07}
{\sc D.~Barth, J.~Cohen, and T.~Faik}, {\em On the $b$-continuity property of
  graphs}, Discrete Applied Mathematics, 155 (2007), pp.~1761--1768.

\bibitem{BHK09}
{\sc A.~Bj{\"o}rklund, T.~Husfeldt, and M.~Koivisto}, {\em Set partitioning via
  inclusion-exclusion}, SIAM Journal on Computing, 39 (2009), pp.~546--563.

\bibitem{CC06}
{\sc M.~Chelb\'{i}k and J.~Chleb\'{i}kova}, {\em Hard coloring problems in low
  degree planar bipartite graphs}, Discrete Applied Mathematics, 154 (2006),
  pp.~1960--1965.

\bibitem{CyganEtAl15}
{\sc M.~Cygan, F.~V. Fomin, L.~Kowalik, D.~Lokshtanov, D.~Marx, M.~Pilipczuk,
  M.~Pilipczuk, and S.~Saurabh}, {\em Parameterized Algorithms}, Springer,
  1st~ed., 2015.

\bibitem{DDJP15}
{\sc K.~K. Dabrowski, F.~Dross, M.~Johnson, and D.~Paulusma}, {\em Filling the
  complexity gaps for colouring planar and bounded degree graphs}, in Proc.
  26th International Workshop on Combinatorial Algorithms (IWOCA '15),
  vol.~9538 of Lecture Notes in Computer Science, Springer, 2015, pp.~100--111.

\bibitem{DF13}
{\sc R.~G. Downey and M.~R. Fellows}, {\em Fundamentals of Parameterized
  Complexity}, Texts in Computer Science, Springer, 2013.

\bibitem{EGT16}
{\sc B.~Effantin, N.~Gastineau, and O.~Togni}, {\em A characterization of
  b-chromatic and partial grundy numbers by induced subgraphs}, Discrete
  Mathematics, 339 (2016), pp.~2157--2167.

\bibitem{GLPR19}
{\sc E.~Galby, P.~T. Lima, D.~Paulusma, and B.~Ries}, {\em On the parameterized
  complexity of $k$-edge colouring}, 2019.
\newblock arXiv:1901.01861.

\bibitem{GJS76}
{\sc M.~R. Garey, D.~S. Johnson, and L.~Stockmeyer}, {\em Some simplified
  {NP}-complete graph problems}, Theoretical Computer Science, 1 (1976),
  pp.~237--267.

\bibitem{HSS12}
{\sc F.~Havet, C.~L. Sales, and L.~Sampaio}, {\em $b$-{C}oloring of tight
  graphs}, Discrete Applied Mathematics, 160 (2012), pp.~2709--2715.

\bibitem{HS13}
{\sc F.~Havet and L.~Sampaio}, {\em On the grundy and b-chromatic numbers of a
  graph}, Algorithmica, 65 (2013), pp.~885--899.

\bibitem{IM99}
{\sc R.~W. Irving and D.~F. Manlove}, {\em The $b$-chromatic number of a
  graph}, Discrete Applied Mathematics, 91 (1999), pp.~127--141.

\bibitem{Kar72}
{\sc R.~M. Karp}, {\em Reducibility among combinatorial problems}, in
  Complexity of computer computations, Springer, 1972, pp.~85--103.

\bibitem{KTV02}
{\sc J.~Kratochv\'{i}l, Z.~Tuza, and M.~Voigt}, {\em On the $b$-chromatic
  number of graphs}, in Proc. 28th International Workshop on Graph-Theoretic
  Concepts in Computer Science (WG '02), vol.~2573 of Lecture Notes in Computer
  Science, 2002, pp.~310--320.

\bibitem{PPS17}
{\sc F.~Panolan, G.~Philip, and S.~Saurabh}, {\em On the parameterized
  complexity of $b$-chromatic number}, Journal of Computer and System Sciences,
  84 (2017), pp.~120--131.
\newblock Previously appeared at IPEC '15.

\bibitem{Sam12}
{\sc L.~Sampaio}, {\em Algorithmic Aspects of Graph Coloring Heuristics}, PhD
  thesis, Universit\'{e} Nice Sophia Antipolis, France, 2012.

\end{thebibliography}

\end{document}